\documentclass[11pt,american]{article}
\usepackage[T1]{fontenc}
\usepackage[utf8]{inputenc}
\usepackage[a4paper]{geometry}
\usepackage{url}
\usepackage{amsmath}
\usepackage{amsthm}
\usepackage{amssymb}
\usepackage{graphicx}

 \renewcommand{\vec}[1]{\boldsymbol{#1}}

\theoremstyle{definition}
\newtheorem{define}{Definiton}[section]
\newtheorem{theorem}[define]{Theorem}
\newtheorem{lemma}[define]{Lemma}

\theoremstyle{remark}

\newcommand{\eps}{\varepsilon}
\newcommand{\divergence}{\mathop{\mathrm{div}}\nolimits}
\newcommand{\grad}{\mathop{\mathrm{grad}}\nolimits}

\usepackage{color}
\usepackage[normalem]{ulem}
\definecolor{DarkGreen}{rgb}{0,0.5,0.1} 

\newcommand\soutD{\bgroup\markoverwith
{\textcolor{DarkGreen}{\rule[.5ex]{2pt}{1pt}}}\ULon}
\newcommand{\Hm}[1]{\leavevmode{\marginpar{\tiny%
$\hbox to 0mm{\hspace*{-0.5mm}$\leftarrow$\hss}%
\vcenter{\vrule depth 0.1mm height 0.1mm width \the\marginparwidth}%
\hbox to
0mm{\hss$\rightarrow$\hspace*{-0.5mm}}$\\\relax\raggedright #1}}}

\begin{document}

\title{\textbf{Effective quantum Hamiltonian in thin domains with non-homogeneity}}
\author{Romana Kvasničková}
\date{\emph{Department of Mathematics, Faculty of Nuclear Sciences and Physical Engineering, Czech Technical University in Prague, Trojanova 13, 120 00 Prague, Czech Republic; kvasnrom@fjfi.cvut.cz}
\medskip \\
{\small March 16, 2022}}

\maketitle

\begin{abstract} 
\noindent
We consider the Laplacian with a non-homogeneous metric
in a tubular neighbourhood of a compact hypersurface
in the Euclidean space of arbitrary dimension,
subject to Neumann boundary conditions.
It is shown that, 
in the limit of the width of the neighbourhood shrinking to zero, 
the operator converges in a generalised
norm-resolvent sense to an effective Laplace--Beltrami-type operator 
on the hypersurface.
In this way, we generalise and give an insight into
the convergence of eigenvalues obtained by Yachimura~\cite{tosy}.
\vspace{0.5cm}

\end{abstract}
%

\section{Introduction}
%

Let $\Sigma$ be a connected orientable compact 
$C^3$ hypersurface in $\mathbb{R}^{d}$, for $d\geq 2$, with the Riemannian metric $g$ induced by the embedding.
The orientation is fixed by the choice of 
a global unit normal vector field $\vec{n}$ on~$\Sigma$.
Since~$\Sigma$ is compact, the tubular neighbourhood  
\begin{equation}\label{domain}
    \Omega_{\varepsilon} :=
    \left\{ s+ \varepsilon t \vec{n} 
    \ \big| \ (s,t)\in \Sigma\times I\right\}
\end{equation}
with $I:=(-1, 1)$
is a $C^2$-smooth open set for all sufficiently small $\varepsilon>0$.  
This paper is concerned with spectral properties of 
the Laplace-type operator
\begin{equation}\label{operator}
  \tilde{H}_\eps  := -\divergence a\grad
  \qquad \mbox{in} \qquad 
  L^2(\Omega_{\varepsilon})
  ,
\end{equation}
subject to Neumann boundary conditions,
where $a:\Omega_{\varepsilon} \to \mathbb{R}$
is a uniformly positive and bounded function.  

The operator~$\tilde{H}_\eps$ can be interpreted as the quantum Hamiltonian
of a (quasi-)particle constrained to a non-homogeneous nanostructure.
It is well known that the effective quantum dynamics in the limit $\eps \to 0$
drastically depends on the choice of boundary conditions 
(see~\cite{K5,K10} for a comparison).
For the present case of Neumann boundary conditions,
the pioneering work of Schatzman \cite{Schatzman_1996}
in the case of the \emph{homogeneous} Laplacian (i.e.\ $a=1$)
proves that the effective dynamics is governed
by the Laplace--Beltrami operator 
$-\Delta_g = \divergence_g \, \grad_g$ on the hypersurface~$\Sigma$.
More specifically, denoting by 
$\{(\lambda_{\varepsilon})_{k}\}_{k=1}^\infty$
and $\{\lambda_{k}\}_{k=1}^\infty$
the set of eigenvalues of~$\tilde{H}_\eps$ and $-\Delta_g$,
respectively, arranged in an increasing order 
and repeated according to multiplicity,
one has the following convergence result.
\begin{theorem}[Schatzman~\cite{Schatzman_1996}]\label{Thm.Schatzman}
If $a(x) = 1$ for almost every $x \in \Omega_\eps$, 
then (for each $k \in \mathbb{N}$)
\begin{equation}\label{Schatzman}
  (\lambda_{\varepsilon})_{k} = \lambda_{k} + O(\eps)
  \qquad \mbox{as} \qquad
  \eps \to 0. 
\end{equation}
\end{theorem}
\noindent The fact that~\eqref{Schatzman} is a consequence of a generalised
norm-resolvent convergence of $\tilde{H}_\eps$ with $a=1$
to~$-\Delta_g$ as $\eps \to 0$ was established in~\cite{KRRS2}.
Actually, the case of general Robin boundary conditions
is considered in that paper. 

The goal of the present paper is to study 
the influence of the non-homogeneity $a \not = 1$
on the effective quantum dynamics in the limit $\eps \to 0$. 
Our primary motivation is the recent paper of
Yachimura~\cite{tosy} in which a special case 
of piece-wise constant non-homogeneity 
in the half-tubular neighbourhoods 
$$
  \Omega_{\varepsilon}^{\pm}
  :=\left\{s \pm \varepsilon t\vec{n}\in \Omega_{\varepsilon} 
  \ | \ t\in (0,1)\right\}
$$
is investigated (see Figure~\ref{fig:test}).
\begin{theorem}[Yachimura~\cite{tosy}]\label{Thm.Yachimura}
If $a(x) := a_\pm$ for almost every $x \in \Omega_{\varepsilon}^{\pm}$
where $a_\pm$ are given positive constants, 
then (for each $k \in \mathbb{N}$)
\begin{equation}\label{tos kon}
   \left(\lambda_{\varepsilon}\right)_{k}
   = \frac{a_{-}+a_{+}}{2}\lambda_{k} + O(\eps)
   \qquad \mbox{as} \qquad
   \eps \to 0.
\end{equation}
\end{theorem}  

As a matter of fact, under the additional hypothesis
that the $k$-th eigenvalue~$\lambda_k$ is simple,
more precise asymptotics are derived 
in~\cite{Schatzman_1996} and~\cite{tosy}.

\begin{figure}[h!]
\centering
  \includegraphics[width=0.4\textwidth]{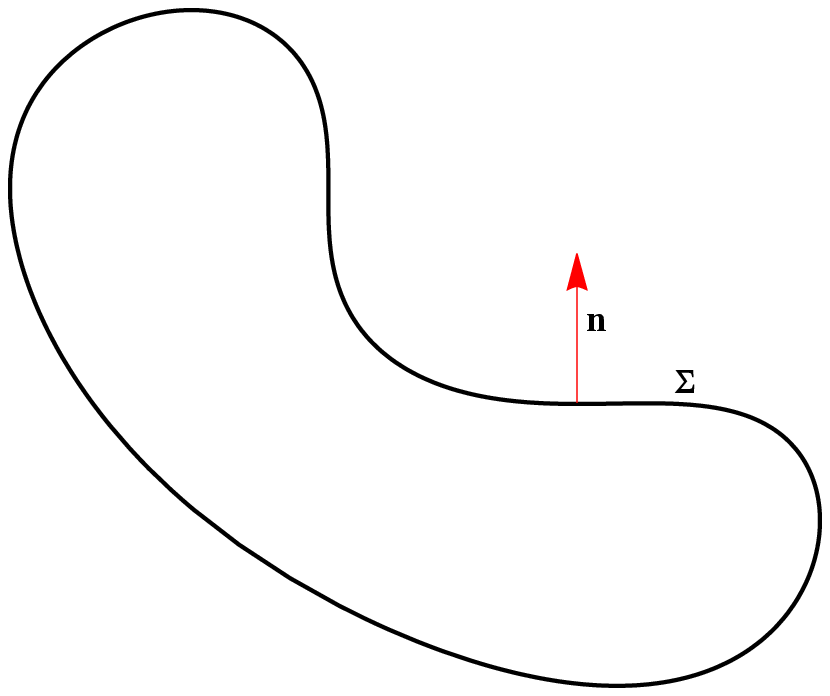}
  \qquad 
  \includegraphics[width=0.48\textwidth]{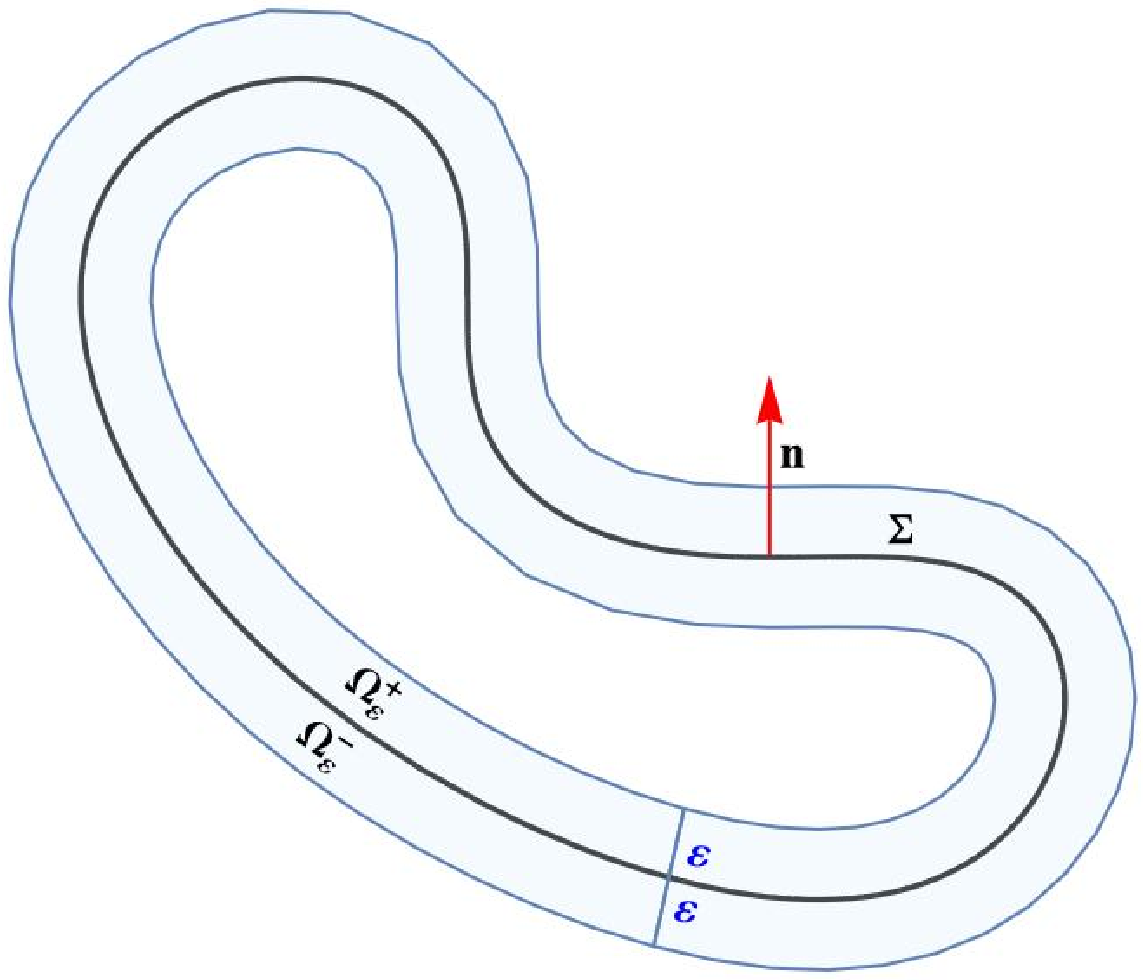}
  \caption{An example of hypersurface $\Sigma$ in $\mathbb{R}^{2}$
and the corresponding tubular neighbourhood $\Omega_{\varepsilon}$.}
\label{fig:test}
\end{figure} 
The principal objective of the present paper is 
to give an insight into the result~\eqref{tos kon}. 
Why the limit eigenvalues are equal to the eigenvalues of 
the Laplace-Beltrami operator multiplified 
by the arithmetic mean of the constants $a_+$ and $a_{-}$?
What is the effective Hamiltonian governing 
the quantum dynamics in the limit $\eps \to 0$?
Second, we wish to understand the convergence of eigenvalues~\eqref{tos kon} 
as a consequence of more robust operator limits.
In particular, we obtain convergence of eigenfunctions as well.
Finally, we go beyond the special piece-wise constant 
scenario of Yachimura's by considering a fully arbitrary non-homogeneity~$a$.

The present paper reveals that the effective dynamics
is given by the operator 
\begin{equation}\label{effective}
   \tilde{H}_{\mathrm{eff}}
   : =- \divergence_g \, \overline{a} \, \grad_g
   \qquad \mbox{in} \qquad L^{2}(\Sigma),
\end{equation}
where (assuming that the limit exists)
\begin{equation}\label{abar}
  \overline{a}(s) := \lim_{\eps \to 0} \ \frac{1}{2}
  \int_{-1}^1 a(s + \varepsilon t\vec{n}) \, \mathrm{d}t.
\end{equation}
In fact, under suitable hypotheses, 
we show that~$\tilde{H}_\eps$ converges to $\tilde{H}_{\mathrm{eff}}$
in a generalised norm-resolvent sense.
This provides an explanation for both 
Theorems~\ref{Thm.Schatzman} and~\ref{Thm.Yachimura}.
Indeed, if the non-homogeneity~$a$ is continuous 
(in particular, if $a=1$), then the effective operator
is given by~\eqref{effective} with~$\overline{a}$ 
being just the projection of~$a$ on the hypersurface~$\Sigma$,
so Theorem~\ref{Thm.Schatzman} follows as a special case. 
In the piece-wise constant scenario of Yachimura's,
it is immediate to see that~$\overline{a}$ is just 
the arithmetic mean of the constants~$a_\pm$,
so Theorem~\ref{Thm.Yachimura} follows as a special case too. 
Of course, our general setting covers much more complicated 
shapes of non-homogeneities, including those with variable~$\overline{a}$.

The generalised norm-resolvent convergence of~$\tilde{H}_\eps$ 
to $\tilde{H}_{\mathrm{eff}}$ is stated as Theorem~\ref{MTHM}
at the end of the paper.
Since the precise statement requires a number of geometric 
and functional preliminaries, 
here we restrict to an informal overview only.
First of all, it is important to deal with the fact 
that the operators~$\tilde{H}_\eps$ and $\tilde{H}_{\mathrm{eff}}$
act on different Hilbert spaces.
Moreover, the former acts on an $\eps$-dependent Hilbert space.
Therefore our first step is to introduce a unitary transform
$$
  \tilde{U}: L^2(\Omega_\eps) \to L^2(\Sigma \times I)
  \,,
$$   
which is realised by a natural change of coordinates
(coming from the very definition~\eqref{domain})
and by subsequently removing the $\eps$-dependent Jacobian.
Our second step is to introduce the isomorphism
\begin{equation}\label{pi}
  \pi: \mathcal{H}_0 := \{ \varphi\otimes  \chi_{0} \ | \ 
  \varphi\in L^{2}(\Sigma) \} \to L^{2}(\Sigma) :
  \{\varphi\otimes  \chi_{0} \to \varphi\} ,
\end{equation}
where $\chi_{0}(t) :=(\sqrt{2})^{-1}$
is the fundamental state of the Neumann Laplacian 
on the one-dimensional interval~$I$.
Then the transformed operator 
$H_{\mathrm{eff}} := \pi^{-1} \tilde{H}_{\mathrm{eff}} \, \pi$
acts in the subspace $\mathcal{H}_0 \subset L^2(\Sigma \times I)$.
With these identifications, Theorem~\ref{MTHM} states that
there exists an $\eps$-independent constant~$C$ such that
\begin{equation}\label{nrs}
  \left\| 
  \tilde{U} (\tilde{H}_{\varepsilon}+1)^{-1} \tilde{U}^{-1}
  -\left(H_{\mathrm{eff}}+1\right)^{-1}\oplus 0^{\perp} 
  \right\| \leq C \max\left\{ \varepsilon, d(\varepsilon)\right\}
\end{equation}
for all sufficiently small~$\eps$.
Here $0^{\perp}$ denotes the zero operator on $\mathcal{H}_0^\bot$ and
$$
  d(\varepsilon) :=
  \underset{s\in \Sigma} {\mathrm{ess}\sup}
  \left| \ 
  \overline{a}(s) - \frac{1}{2}
  \int_{-1}^1 a(s + \varepsilon t\vec{n}) \, \mathrm{d}t
  \ \right|
  .
$$
Since $d(\eps) = O(\eps)$ in the scenarios of both
Theorems~\ref{Thm.Schatzman} and~\ref{Thm.Yachimura},
the norm-resolvent convergence catches the correct decay rate 
in the known convergence of eigenvalues~\eqref{Schatzman} and~\eqref{tos kon}.
What is more, \eqref{nrs}~implies the convergence eigenprojections
(and therefore of suitably normalised eigenfunctions,
see~\cite[Sec.~IV.3.5]{kato}) 
with the same decay rate $O(\eps)$.
We leave as an open problem whether~$d(\eps)$ 
might be the optimal decay rate in less regular scenarios. 

The paper is organised as follows. 
In the forthcoming Section~\ref{GP}, 
we develop geometric preliminaries about the hypersurface~$\Sigma$ 
and its tubular neighbourhoods~$\Omega_\eps$. 
The following Section~\ref{RF} is concerned with the rigorous 
definition of the Hamiltonians~$\tilde{H}_\eps$ and~$\tilde{H}_{\mathrm{eff}}$
as self-adjoint opera-tors and with unitarily equivalent variants of the former.
The norm-resolvent convergence~\eqref{nrs}
is established in Section~\ref{NRC}.
The paper is concluded by Section~\ref{Sec.end},
where we make a couple of remarks about the optimality of our hypotheses.

\section{Geometric preliminaries}\label{GP}
We start with some standard geometric preliminaries following~\cite{kre-tus},
where a similar geometric setting can be found.

Recall that~$\Sigma$ is assumed to be a connected 
orientable compact $C^3$ hypersurface 
embedded in the Euclidean space~$\mathbb{R}^d$ 
of dimension $d \geq 2$. The orientation can be determined by 
a $C^2$-smooth unit normal vector field $\Vec{n}: \Sigma\rightarrow S^{d-1}$ along the hypersurface $\Sigma$. 
Let the Weingarten map be introduced as
$$
  L: T_{s}\Sigma \longrightarrow T_{s}\Sigma : 
  \left\{ 
  \xi\mapsto -\mathrm{d} \vec{n}(\xi)
  \right\} 
$$
for any $s\in\Sigma$. Denote the local coordinate system of $\Sigma$ by $s^{1}, \dots, s^{d-1}$. The map $L$ can be interpreted as a $(1,1)$ mixed tensor with the matrix representation with respect to the basis $\left(\frac{\partial}{\partial s^{1}}, \dots ,\frac{\partial}{\partial s^{d-1}} \right)$. Henceforth, the range of Greek indices is $1,\dots,d-1$, the range of Latin indices is $1,\dots, d$ and the Einstein summation convention is used. Define the second fundamental form 
$h(\xi,\eta) :=\left< L(\xi),\eta \right>$ for any $\xi, \eta \in T_{s}\Sigma$ and $s\in \Sigma$, where $\left< \cdot,\cdot\right>$ is the scalar product 
determined by the induced metric~$g$.
Then the formula $L^{\mu}_{\nu}=g^{\mu \rho}h_{\rho \nu}$ is valid with the notation $( g^{\mu \nu}) = ( g_{\mu \nu})^{-1}$.  
See, e.g., \cite[Chap.~3]{spivak}
for more details on the present geometric discussion.

The eigenvalues of $L$ are known as the principal curvatures $\kappa_{1},\dots, \kappa_{d-1}$ of $\Sigma$. We introduce the positive constant 
$$
  \rho 
  := \left( \max_{\mu} \left\{ 
  \| \kappa_{\mu}\|_{L^{\infty}(\Sigma)} 
  \right\}\right)^{-1}
  ,
$$
which is well defined because~$\Sigma$ is assumed 
to be sufficiently regular and compact. 
The mean curvatures 
$$
  K_{\mu} := 
  \binom{d-1}{\mu}^{-1}\sum_{\alpha_{1}< \dots < \alpha_{\mu}}\kappa_{\alpha_{1}}\dots \kappa_{\alpha_{\mu}}
$$
are invariants of~$L$, see \cite[Sec.~3F]{kuhnel}). 
Furthermore, 
while the principal curvatures are \emph{a priori} defined only locally,
the mean curvatures $K_{\mu}$ 
are globally defined $C^{1}$-smooth functions.

Consider the Cartesian-product manifold 
$\Omega := \Sigma \times I$ 
and the $C^{2}$-smooth mapping  
\begin{equation}
    \label{L}
    \mathcal{L}_\varepsilon :\Omega\longrightarrow \mathbb{R}^{d} : \left\{(s,t)\mapsto s+\varepsilon t \vec{n}  \right\}.
\end{equation}
Then define $\Omega_{\varepsilon} := \mathcal{L}_{\varepsilon}(\Omega)$. 
The mapping $\mathcal{L}_{\varepsilon}$ induces the metric~$G$, 
which can be written in a block form
$$G=g\circ (Id-\varepsilon t L)^{2}+\varepsilon^{2}\mathrm{d}t^{2},$$
where $Id$ is the identity on $T_{s}\Sigma$. 
Consequently, 
$
  \left| G\right| := \mathrm{det}(G)
  =\varepsilon^{2}\left| g\right| f_\eps^2 
$,
where $\left| g \right| := \mathrm{det}(g)$ and
$$
  f_\eps(\cdot,t) := 
  \mathrm{det}(1-\varepsilon t L)
  = \prod_{\mu=1}^{d-1}(1-\varepsilon t \kappa_{\mu})
  = 1+\sum_{\mu=1}^{d-1}(-\varepsilon t)^{\mu} \, 
  \binom{d-1}{\mu} \, K_{\mu} .
$$

By elementary estimates,
for every $(s,t) \in \Omega$,
\begin{multline}\label{crucial}
  c_{\varepsilon}^{-}:=
  2-\left(1+\frac{\varepsilon}{\rho}\right)^{d-1}
  = 
  1 - \sum_{\mu=1}^{d-1}\left(\frac{\varepsilon }{\rho}\right)^{\mu}
  \binom{d-1}{\mu}
  \\
  \leq f_\eps(s,t) \leq
  \\
  1 + \sum_{\mu=1}^{d-1}\left(\frac{\varepsilon }{\rho}\right)^{\mu}
  \binom{d-1}{\mu}=\left(1+\frac{\varepsilon}{\rho}\right)^{d-1}
  =: c_{\varepsilon}^{+} .
\end{multline}
It follows that~$c_-^\eps$ is positive and~$c_+^\eps$ bounded,
both uniformly in~$\eps$, provided that
\begin{equation}\label{overlap1}
    \varepsilon<\rho  
    .
\end{equation}
Consequently, $(\Omega,G)$ is an immersed Riemannian $C^2$ manifold, 
or in other words, $\mathcal{L}_{\varepsilon}: \Omega \to \Omega_\eps$ 
is a local $C^2$-diffeomorphism. 
To make it a global $C^2$-diffeomorphism, 
we additionally assume that
\begin{equation}\label{overlap2}
  \mathcal{L}_{\varepsilon} 
  \hspace{0.2cm}\mathrm{is}\hspace{0.1cm}\mathrm{injective}.
\end{equation}
Then $(\Omega,G)$ is an embedded Riemannian $C^2$ manifold.
In particular, $\Omega_\eps$ is an open set of class~$C^2$.

Obviously, \eqref{overlap1}~is always satisfied 
for all sufficiently small~$\eps$. 
This is also true for~\eqref{overlap2}
because we assume that~$\Sigma$ is compact. 
From now on
(without loss of generality because we are interested 
in the regime of small~$\eps$ anyway),
we shall therefore assume that~$\eps$
is so small that the two conditions~\eqref{overlap1}
and~\eqref{overlap2} are met.
  
Writing $x^{d}:=t$ and 
$\frac{\partial}{\partial x^{d}}:=\frac{\partial}{\partial t}$ for $t\in I$, the metric $G$ can be represented by the matrix
$$
\left(G_{ij} \right) = \left(\begin{matrix} \left(G_{\mu \nu} \right) & 0 \\ 0 & \varepsilon^{2} \end{matrix}\right)
, \hspace{0.8cm} 
G_{\mu \nu}(\cdot,t)=g_{\mu \rho}\left( \delta^{\rho}_{\sigma}-\varepsilon t L^{\rho}_{\sigma}\right)\left( \delta^{\sigma}_{\nu}-\varepsilon t L^{\sigma}_{\nu}\right). $$
Consequently,
\begin{equation}\label{G}
    C_{\varepsilon}^{-} \left(g^{\mu \nu} \right) \leq \left(G^{\mu \nu} \right) \leq C_{\varepsilon}^{+} \left(g^{\mu \nu}\right), \hspace{0.8cm} 
    C_{\varepsilon}^{\pm}:=\left(1\mp \varepsilon\rho^{-1}\right)^{-2},
\end{equation}
where the notation $(G^{\mu \nu})=(G_{\mu \nu} )^{-1}$ is used. Moreover, the assumptions about the hypersurface $\Sigma$ give us an positive constant $B$ such that the relations
\begin{equation}\label{metrika}
    B^{-1} \left(\delta^{\mu \nu} \right) \leq \left(g^{\mu \nu} \right) \leq B \left(\delta^{\mu \nu}\right), \hspace{0.8cm} \left| \frac{\partial g^{\mu\nu}}{\partial s^{\rho}}\right| \leq B
\end{equation}
are valid.

Finally, 
on the Riemannian manifold $(\Omega,G)$,
we introduce the volume element 
\begin{equation}\label{element}
\left| G\right|^{1/2}\left| g\right|^{-1/2} \mathrm{d}\Sigma\wedge\mathrm{d}t
= \varepsilon f_{\varepsilon}\hspace{0.1cm} \mathrm{d}\Sigma\wedge\mathrm{d}t,
\end{equation}
where $ \mathrm{d}\Sigma = \left| g\right|^{1/2} \mathrm{d}s^{1}\wedge\dots \wedge \mathrm{d}s^{d-1}$.

\section{Functional preliminaries}\label{RF}
Since the configuration space~$\Omega_\eps$ is $\eps$-dependent,
the non-homogeneity $a \in L^\infty(\Omega_\eps)$ is allowed 
to be $\eps$-dependent, too.   
However, we always make the following uniform hypothesis:
there exists a positive constant~$c$ such that 
\begin{equation}\label{a0}
  c\leq a(x)< c^{-1}
\end{equation}
for almost every $x \in \Omega_\eps$
and for all sufficiently small~$\eps$.

The operator $\tilde{H}_{\varepsilon}$ formally introduced in~\eqref{operator}
is rigorously defined as the self-adjoint operator associated
in the Hilbert space $L^2(\Omega_\eps)$
with the quadratic form
$$
  \tilde{Q}_{\varepsilon}[\tilde{v}] :=
\int_{\Omega_{\varepsilon}} a(x) \, \left|\nabla \tilde{v}(x)\right|^{2}
  \, \mathrm{d} x 
  , \qquad
  D(\tilde{Q}_{\varepsilon}) :=W^{1,2}(\Omega_{\varepsilon}).
$$
This is well defined due to~\eqref{a0}
and because~$\Omega_{\varepsilon}$ is an open set 
under our geometric hypotheses~\eqref{overlap1} and~\eqref{overlap2}.   

Define the unitary transformation
\begin{equation*}
    \label{U}
    U:L^{2}\left(\Omega_{\varepsilon}\right)
    \to
    L^{2}\big(\Omega,f_{\varepsilon}\,\mathrm{d}\Sigma\wedge\mathrm{d}t\big): 
    \left\{
    \tilde{v} \mapsto v=\sqrt{\varepsilon} \ \tilde{v}\circ \mathcal{L}_{\varepsilon}
    \right\},
\end{equation*}
where $\mathcal{L}_{\varepsilon}$ is given by (\ref{L}). 
We emphasise the presence of the $\eps$-dependent scaling factor,
which makes the $\eps$-dependence of the target space 
concentrated in the function~$f_\eps$ only.  
Then we introduce a unitarily equivalent operator
$H_{\varepsilon} :=U\tilde{H}_{\varepsilon}U^{-1}$
in $L^{2}\big(\Omega,f_{\varepsilon}\,\mathrm{d}\Sigma\wedge\mathrm{d}t\big)$.
By definition, $H_{\varepsilon}$~is associated 
with the quadratic form $Q_\eps[v] := \tilde{Q}_\eps[U^{-1}v]$,
$D(Q_\eps) := U D(\tilde{Q}_\eps)$.
It is straightforward to verify that 
$$
  Q_\eps[v] 
  = \int_{\Omega} a_{\varepsilon}\left[ \left(\frac{\partial \overline{v}}{\partial s^{\mu}}G^{\mu \nu}\frac{\partial v}{\partial s^{\nu}} \right) + \frac{1}{\varepsilon^{2}} \left|{\frac{\partial v}{\partial t}}\right|^2\right] f_{\varepsilon}\hspace{0.1cm} \mathrm{d}\Sigma\wedge\mathrm{d}t,
$$
where $a_{\varepsilon} := a\circ\mathcal{L}_{\varepsilon}$.
By~\eqref{a0}, one has 
\begin{equation}\label{a1}
   c \leq a_\eps(s,t) \leq c^{-1}
\end{equation}
for almost every $(s,t) \in \Omega$
and for all sufficiently small~$\eps$.
Moreover, 
$$
  D(Q_\eps) =
  W^{1,2}\big(\Omega, f_{\varepsilon} \, \mathrm{d}\Sigma\wedge\mathrm{d}t \big)
  = W^{1,2}(\Omega) ,
$$
where the second equality employs~\eqref{crucial} and~\eqref{G}. 

Similarly, the Hilbert spaces 
$L^{2}\big(\Omega,f_{\varepsilon}\,\mathrm{d}\Sigma\wedge\mathrm{d}t\big)$
and $L^2(\Omega)$ can be identified as vector spaces.
However, the respective norms, denoted here by  
$\| \cdot \|_{\eps}$ and $\| \cdot \|_{0}$,
are different (though equivalent).
More specifically, 
\begin{equation}\label{topology}
  c_\eps^- \, \| v \|_{0}
  \leq \| v \|_{\eps} \leq c_\eps^+ \, \| v \|_{0}
\end{equation}
for every $v \in L^2(\Omega)$, 
where~$c_\eps^\pm$ are the constants from~\eqref{crucial}.
Notice that $c_\eps^\pm \to 1$ as $\eps \to 0$.

The effective operator $\tilde{H}_{\mathrm{eff}}$ 
formally introduced in~\eqref{effective}
is rigorously defined as the self-adjoint 
operator associated in the Hilbert space $L^2(\Sigma)$ 
with the quadratic form
$$
  \tilde{Q}_{\mathrm{eff}}[\varphi]
  :=\int_{\Sigma} \overline{a} \ 
  | \nabla_{\!g} \, \varphi |_{g}^{2} \, \mathrm{d}\Sigma
  , \qquad
  D(\tilde{Q}_{\mathrm{eff}}) := W^{1,2}(\Sigma) .
$$
Here we employ the notation 
$| \nabla_{\!g} \, \varphi |_{g}
:= \sqrt{\frac{\partial \varphi}{\partial s^{\mu}}g^{\mu \nu}\frac{\partial \varphi}{\partial s^{\nu}}}$ 
and assume that the limit~\eqref{abar} exists
(sufficient conditions will be discussed later on).
Note that 
\begin{equation}\label{d0}
  \overline{a}
  :=\lim_{\varepsilon\mapsto 0}\frac{\left<a_{\varepsilon} \right>}{2}
  \qquad \mbox{with} \qquad
  \left<a_{\varepsilon} \right\rangle (s)
  :=\int_{-1}^{1} a_{\varepsilon}(s,t) \, f_{\varepsilon}(s,t) \, \mathrm{d} t
  .
\end{equation} 

We shall need to identify $\tilde{H}_{\mathrm{eff}}$  
with an operator in $L^2(\Omega)$. We will conduct it similarly as in \cite[Sect. 5.1]{sed}.
To this aim, 
let $\left\{\chi_{n}\right\}_{n=0}^{\infty}$ 
be the orthonormal basis in $L^{2}(I)$ 
given by the eigenfunctions of the associated Neumann Laplacian:
\begin{equation}\label{zz2}
  \chi_{n}(t):=
  \begin{cases}
  \frac{1}{\sqrt{2}} & \mathrm{for }\hspace{0,2 cm} n=0,
  \\ 
  \cos\left(\frac{n\pi}{2}t\right) & \mathrm{for }\hspace{0,2 cm} n\geq 1 \hspace{0,2cm} \mathrm{ even},
  \\
\sin\left(\frac{n\pi}{2}t\right) & \mathrm{for }\hspace{0,2 cm} n \hspace{0,2cm}\mathrm{odd}.
\end{cases}
\end{equation}
Then any $\psi\in L^{2}\left(\Omega \right)$ can be expanded into the Fourier series 
$$\psi(s,t)=\sum_{n=0}^{\infty}\psi_{n}(s)\chi_{n}(t)=\psi_{0}(s)\chi_{0}(t)+\sum_{n=1}^{\infty}\psi_{n}(s)\chi_{n}(t),$$ 
were $\psi_{n}(s):=\left(\chi_{n},\psi(s,.)\right)_{2}$ 
for any $n\in\mathbb{N}_{0}$.
Here $(\cdot,\cdot)_2$ denotes the inner product in $L^2(I)$;
the associated norm will be denoted by~$\|\cdot\|_2$.
Let $\mathcal{H}_{0}$ be the (closed) 
subspace of $L^2(\Omega)$ defined in~\eqref{pi}.
The operator  
 \begin{equation}
\label{P1}
P_{0}: L^{2}(\Omega)\rightarrow L^{2}(\Omega) : \left\{ \psi \mapsto \psi_{0}\otimes\chi_{0} 
=\left(\int_{-1}^{1}\psi(\cdot,t)\hspace{0.1cm} \mathrm{d}t\right)\otimes\frac{1}{2} \right\}
\end{equation}
realises the orthogonal projection onto $\mathcal{H}_{0}$.
Consequently, one has the orthogonal decomposition
$$
  L^{2}(\Omega)=P_{0}\left(L^{2}(\Omega)\right)\oplus P_{0}^{\perp}\left(L^{2}(\Omega)\right)=\mathcal{H}_{0}\oplus \mathcal{H}_{0}^{\perp}
  ,
$$ 
where $P_{0}^{\perp}:=1-P_{0}$. Additionally, the map $\pi$ defined in~\eqref{pi}
is an isometric isomorphism. 
Then the operator 
$H_{\mathrm{eff}} := \pi^{-1} \tilde{H}_\mathrm{eff} \, \pi$
acts in $\mathcal{H}_{0} \subset L^2(\Omega)$.
The associated quadratic form reads
$$
  Q_{\mathrm{eff}}\left[\varphi\otimes \chi_{0}\right]
  :=  \int_{\Omega}
  \left( \overline{a} \ | \nabla_{\!g}\,\varphi|_{g}^2 \right)
  \otimes |\chi_{0}|^2 \
  \mathrm{d}\Sigma \wedge\mathrm{d}t = \tilde{Q}_{\mathrm{eff}}[\varphi]
  , \qquad
  D\left(Q_{\mathrm{eff}}\right) :=
  \left\{ \varphi\otimes \chi_{0} \ | \ 
  \varphi\in W^{1,2}(\Sigma)
  \right\}.
$$

\section{The norm-resolvent convergence}\label{NRC}

\subsection{Resolvent equations}
Let us introduce yet another unitary transformation: 
\begin{equation} 
    \label{Ue}
     U_{\varepsilon}:L^{2}\left(\Omega,f_{\varepsilon}\hspace{0.1cm}\mathrm{d}\Sigma\wedge\mathrm{d}t \right)\to L^{2}\left(\Omega\right): \left\{v \mapsto  f_{\varepsilon}^{1/2} \, v\right\}.
\end{equation}
Given arbitrary functions 
$F,G\in L^{2}\left(\Omega\right)$,
we introduce $\left\{\psi_{\varepsilon}\right\}_{\varepsilon> 0}\in L^{2}\left(\Omega, f_{\varepsilon}\,\mathrm{d}\Sigma\wedge\mathrm{d}t\right)$ 
and $\psi\in\mathcal{H}_{0}$ 
as solutions of the following resolvent equations:
\begin{align}
    \left(H_{\varepsilon}+1\right)\psi_{\varepsilon}&=U_{\varepsilon}^{-1}G,
    \label{c}
    \\
    \left(H_{\mathrm{eff}}+1\right)\psi&=P_{0}F,
    \label{d}
\end{align}
where $P_{0}$ is defined in~(\ref{P1}). 
The equation~(\ref{c}) precisely means 
$
  Q_{\varepsilon}(v,\psi_{\varepsilon})
  +(v, \psi_{\varepsilon})_{\varepsilon}
  =(v, U_{\varepsilon}^{-1}G)_{\varepsilon}
$
for any test function $v \in L^{2}\left(\Omega, f_{\varepsilon}\,\mathrm{d}\Sigma\wedge\mathrm{d}t\right)$.
Explicitly,
\begin{equation}
    \label{rr}
    \int_{\Omega} a_{\varepsilon}\left[ \frac{\partial \overline{v}}{\partial s^{\mu}}G^{\mu \nu}\frac{\partial \psi_\varepsilon}{\partial s^{\nu}} + \frac{1}{\varepsilon^{2}} {\frac{\partial \overline{v}}{\partial t}}{\frac{\partial \psi_{\varepsilon}}{\partial t}}\right] f_{\varepsilon}\hspace{0.1cm} \mathrm{d}\Sigma\wedge\mathrm{d}t+\int_{\Omega}\overline{v}\psi_{\varepsilon}f_{\varepsilon}\hspace{0.1cm} \mathrm{d}\Sigma\wedge\mathrm{d}t=\int_{\Omega}\overline{v}G f_{\varepsilon}^{1/2}\hspace{0.1cm} \mathrm{d}\Sigma\wedge\mathrm{d}t.
\end{equation}
Setting $v=\psi_{\varepsilon}$,
\begin{equation}
    \label{r1}
     \int_{\Omega} a_{\varepsilon}\left[ \frac{\partial \overline{\psi_\varepsilon}}{\partial s^{\mu}}G^{\mu \nu}\frac{\partial \psi_\varepsilon}{\partial s^{\nu}} + \frac{1}{\varepsilon^{2}} \left|{\frac{\partial \psi_{\varepsilon}}{\partial t}}\right|^2\right] f_{\varepsilon}\hspace{0.1cm} \mathrm{d}\Sigma\wedge\mathrm{d}t+\int_{\Omega}\left|\psi_{\varepsilon}\right|^{2}f_{\varepsilon}\hspace{0.1cm} \mathrm{d}\Sigma\wedge\mathrm{d}t=\int_{\Omega}\overline{\psi_{\varepsilon}}G f_{\varepsilon}^{1/2}\hspace{0.1cm} \mathrm{d}\Sigma\wedge\mathrm{d}t.
\end{equation}
Analogously, the following equality can be obtained from (\ref{d}):
\begin{equation}.
    \label{f}
    \int_{\Omega} \overline{a} \frac{\partial \overline{\psi}}{\partial s^{\mu}}g^{\mu \nu}\frac{\partial \psi}{\partial s^{\nu}} \hspace{0.1cm} \mathrm{d}\Sigma\wedge\mathrm{d}t+\int_{\Omega} \left| \psi\right|^2 \mathrm{d}\Sigma\wedge\mathrm{d}t=\int_{\Omega} \overline{\psi}P_{0}F\hspace{0.1cm} \mathrm{d}\Sigma\wedge\mathrm{d}t
\end{equation}

\subsection{Resolvent estimates}\label{EF}
Henceforth, $C$~denotes a generic positive constant independent of~$\eps$
whose value can change from line to line and $\left\| \cdot \right\|_{\infty}$ stands for $\left\| \cdot \right\|_{L^{\infty}(\Omega)}$.
Moreover, in the estimates below we use
the common notation $\|\cdot\|$ for $\|\cdot\|_\eps$ or $\|\cdot\|_0$;
this is justified due to~\eqref{topology}.

\begin{lemma}
\label{odhady}
If $\left\{\psi_{\varepsilon}\right\}_{\varepsilon> 0}$ is defined by (\ref{c}) for an arbitrary $G\in L^{2}\left(\Omega\right)$, 
then 
\begin{align}
 \left\|\psi_{\varepsilon}\right\| 
 &\leq C \left\|G\right\| ,
 &
 \left\| \left| \nabla_{g}\psi_{\varepsilon} \right|_{g} \right\| 
 &\leq C \left\| G\right\| ,
 \label{0} 
 \\
 \left\| \frac{\partial P_{0}^{\perp}\psi_{\varepsilon}}{\partial t}\right\|
 =\left\| \frac{\partial\psi_{\varepsilon}}{\partial t}\right\| 
 &\leq C \varepsilon\left\| G\right\| ,
 &
 \left\| \left| 
 \nabla_{g}\left(P_{0}^{\perp}\psi_{\varepsilon}\right) \right|_{g} \right\| 
 &\leq C \left\| G\right\| ,
 \label{2}
 \\
 \left\| P_{0}\psi_{\varepsilon}\right\| 
 &\leq C \left\|G\right\| ,
 &
 \left\| \left| \nabla_{g}\left(P_{0}\psi_{\varepsilon}\right) 
 \right|_{g} \right\|    
 & \leq C \left\| G\right\| ,
\end{align}
for all sufficiently small~$\eps$.
If $\psi$ is defined by (\ref{d}) for an arbitrary $F\in L^{2}\left(\Omega\right)$, then 
\begin{equation}
\label{4}
\left\|\psi\right\|\leq C \left\| F\right\|,
\hspace{1.5cm}\left\| \left| \nabla_{g}\psi \right|_{g} \right\| \leq C \left\| F\right\|.
\end{equation}
\end{lemma}
\begin{proof}
Consider the identity (\ref{r1}). 
From~(\ref{G}), where the constants~$C_\eps^\pm$
can be made $\eps$-independent if~$\eps$ is assumed to be less
than a sufficiently small constant, 
we deduce
$$
 C \int_{\Omega} \left[  \left| \nabla_{g}\psi_{\varepsilon} \right|_{g}^{2} + \frac{1}{\varepsilon^{2}} \left|{\frac{\partial \psi_{\varepsilon}}{\partial t}}\right|^2\right] f_{\varepsilon}\hspace{0.1cm} \mathrm{d}\Sigma\wedge\mathrm{d}t+\int_{\Omega} \left| \psi_{\varepsilon}\right|^2f_{\varepsilon}\hspace{0.1cm} \mathrm{d}\Sigma\wedge\mathrm{d}t\leq\int_{\Omega} \overline{\psi_{\varepsilon}}Gf^{1/2}_{\varepsilon}\hspace{0.1cm} \mathrm{d}\Sigma\wedge\mathrm{d}t. 
$$
Using the Cauchy--Schwarz inequality and~\eqref{crucial},
we obtain  
$$C \left( \left\| \left| \nabla_{g}\psi_{\varepsilon} \right|_{g} \right\|^{2}+ \frac{1}{\varepsilon^{2}}\left\| \frac{\partial\psi_{\varepsilon}}{\partial t}\right\|^{2}  + \left\|\psi_{\varepsilon}\right\|^{2}\right)\leq \left\|\psi_{\varepsilon}\right\|\left\|G\right\|.$$
It implies (\ref{0}) and the first part of (\ref{2}). 
Similarly, estimates in (\ref{4}) follow from (\ref{f}). 

Due to the orthonormality of the basis 
$\left\{\chi_{n}\right\}_{n=0}^{\infty}$ given by (\ref{zz2}), 
the function $P_{0}^{\perp}\psi_{\varepsilon}$ satisfies  
\begin{equation*}
    \label{kkk}
    \int_{-1}^{1} P_{0}^{\perp}\psi_{\varepsilon}\hspace{0.05cm}\mathrm{d}t=0 ,
     \hspace{2 cm}\int_{-1}^{1} \frac{\partial P_{0}^{\perp}\psi_{\varepsilon}}{\partial s^{\mu}}\hspace{0.05cm}\mathrm{d}t=0,
\end{equation*}
where the second equality holds for any $\mu$.
As a consequence, the equalities 
$$
\begin{aligned}
\left\| \psi_{\varepsilon}\right\|_{0} ^2
&=\int_{\Omega}\left[ \left| P_{0}\psi_{\varepsilon}\right| ^2 + \left| P_{0}^{\perp}\psi_{\varepsilon}\right| ^2 \right] \mathrm{d}\Sigma\wedge\mathrm{d}t +2\mathrm{Re}\int_{\Sigma}P_{0}\overline{\psi}_{\varepsilon} \left(\int_{-1}^{1} P_{0}^{\perp} \psi_{\varepsilon} \mathrm{d}t\right)\mathrm{d}\Sigma 
\\
&=\left\| P_{0}\psi_{\varepsilon}\right\|_{0} ^2+\left\| P_{0}^{\perp} \psi_{\varepsilon}\right\|_{0} ^2
\end{aligned}
$$
and 
$$
\begin{aligned}
  \left\| \left| \nabla_{g}\psi_{\varepsilon}\right|_{g}\right\|_{0} ^2
  &=\int_{\Omega} \left[ \left| \nabla_{g}\left(P_{0}\psi_{\varepsilon}\right) \right|_{g}^{2}+\left| \nabla_{g}\left(P_{0}^{\perp} \psi_{\varepsilon}\right)\right|_{g}^{2} \right] \mathrm{d}\Sigma\wedge\mathrm{d}t \
  \\
  &\quad +2\mathrm{Re}\int_{\Sigma}\frac{\partial P_{0}\overline{\psi_\varepsilon}}{\partial s^{\mu}}g^{\mu \nu} \left(\int_{-1}^{1} \frac{\partial P_{0}^{\perp} \psi_\varepsilon}{\partial s^{\nu}} \mathrm{d}t\right)\mathrm{d}\Sigma 
  \\
  &=\left\| \left| \nabla_{g}\left(P_{0}\psi_{\varepsilon}\right) \right|_{g}\right\|_{0} ^2+\left\| \left| \nabla_{g}\left(P_{0}^{\perp} \psi_{\varepsilon}\right)\right|_{g}\right\|_{0} ^2
\end{aligned} 
$$
follow. That is, the remaining estimates of $\psi_{\varepsilon}$ 
have been shown.
\end{proof}

The crucial observation is that the part of~$\psi_\eps$ lying 
in the orthogonal complement~$\mathcal{H}_0^\bot$
vanishes as $\eps \to 0$.

\begin{lemma}\label{ch}
If $\left\{\psi_{\varepsilon}\right\}_{\varepsilon> 0}$ is defined by (\ref{c}) for an arbitrary $G\in L^{2}\left(\Omega\right)$, 
then 
\begin{equation*}
    \label{n4}
    \left\|P_{0}^{\perp}\psi_{\varepsilon}\right\| \leq C \varepsilon\left\| G\right\|
\end{equation*}
for all sufficiently small~$\eps$.
\end{lemma}
\begin{proof}
Recall that the Neumann Laplacian $-\Delta_{N}^{I}$ in $L^2(I)$
is associated with the quadratic form
$$
  Q[\varphi]:=\int_{-1}^{1}\left| \frac{\mathrm{d}\varphi}{\mathrm{d}t} \right| ^2 \mathrm{d}t, \qquad
  D(Q):=W^{1,2}(I).
$$
Eigenvalues of $-\Delta_{N}^{I}$ are 
 $\mu_{n}:=\left(\frac{n\pi}{2}\right)^{2}$ for $n \in \mathbb{N}_0$
 and the corresponding eigenfunctions $\chi_{n}$ are defined in~(\ref{zz2}).
   The minimax principle for $n=1$ can be equivalently written as
\begin{equation}
    \label{mu01}
    \mu_{1}=\underset{\varphi\perp\chi_{0}}{\min_{\varphi\in W^{1,2}(I)}}\frac{Q[\varphi]}{\left\|\varphi\right\|_{2}^{2}}.
\end{equation}
 Consequently, (\ref{mu01}) implies 
\begin{equation}
    \label{mu1}
    \mu_{1}\leq \frac{\int_{-1}^{1}\left| \frac{\mathrm{d}\varphi}{\mathrm{d}t} \right| ^2 \mathrm{d}t}{\int_{-1}^{1}\left|\varphi\right| ^2 \mathrm{d}t}
\end{equation}
for any $\varphi\in W^{1,2}(I)$ satisfying $\varphi\perp \chi_{0}$ 
in $L^{2}(I)$. Finally, the relations 
$$\left\| P_{0}^{\perp}\psi_{\varepsilon}\right\|_{0}^{2}=\int_{\Omega}\left| P_{0}^{\perp}\psi_{\varepsilon}\right| ^2 \mathrm{d}\Sigma\wedge\mathrm{d}t \leq C\int_{\Omega}\left|{\frac{\partial P_{0}^{\perp}\psi_{\varepsilon}}{\partial t}}\right| ^2 \mathrm{d}\Sigma\wedge\mathrm{d}t=C\left\|{\frac{\partial P_{0}^{\perp}\psi_{\varepsilon}}{\partial t}}\right\|_{0}^{2}\leq C \varepsilon^{2}\left\|G\right\|^{2}_{0}$$
can be obtained from (\ref{mu1})
and the first part of~\eqref{2}.
\end{proof}
\begin{lemma}\label{nabla}
Let $a_{\varepsilon} \in L^\infty(\Omega)$ 
satisfy in addition to~(\ref{a1}) the following hypothesis:
there exists a positive constant~$D$ such that 
\begin{equation}\label{a3}
   \left| \nabla_{\!g} \, a_{\varepsilon}\right|_{g} \leq D
\end{equation}
for almost every $(s,t) \in \Omega$
and all sufficiently small~$\eps$. If $\psi$ is defined by (\ref{d}) for an arbitrary $F\in L^{2}(\Omega)$, then 
\begin{equation}
    \left\|  \Delta_{g}\psi \right\| \leq C \left\|F\right\|
\label{delta}
\end{equation}
for all sufficiently small $\varepsilon$. Here the notation $\Delta_{g}\psi:=\left|g \right| ^{-1/2} \frac{\partial}{\partial s^{\mu}}\left( \left|g \right| ^{1/2}\hspace{0.1cm} g^{\mu \nu}\frac{\partial \psi}{\partial s^{\nu}} \right)$ is engaged.
\end{lemma}
\begin{proof}
The proof follows the idea of \cite[proof~of~Lemma~3.3]{deltakrej}. First, the Representation Theorem (cf. \cite[Thm.~VI.1.27]{kato}) directly implies $H_{\mathrm{eff}}\psi\in \mathcal{H}_{0}$. Moreover, both $H_{\mathrm{eff}}\psi=-\overline{a}\Delta_{g} \psi-\frac{\partial\overline{a}}{\partial s^{\mu}}g^{\mu\nu}\frac{\partial\psi}{\partial s^{\nu}}$ and 
$$
\int_{\Sigma} \left| \frac{\partial\overline{a}}{\partial s^{\mu}}g^{\mu\nu}\frac{\partial\psi}{\partial s^{\nu}}\right|^2 \mathrm{d} \Sigma  \leq \| \left| \nabla_{g}\overline{a}\right|_{g}\|^{2}_{L^{\infty}(\Sigma)} \| \left| \nabla_{g}\psi\right|_{g}\|^{2}_{2} \leq C \| \left| \nabla_{g}\psi\right|_{g}\|^{2}_{2}
$$
are satisfied based on the assumptions. Consequently, $\Delta_{g}\psi\in \mathcal{H}_{0}$.\\
\indent Consider the test function $-\partial^{-h}_{\rho}\partial^{h}_{\rho}\psi$, where
$$
\partial^{h}_{\rho}\psi(s,t):=\frac{\psi(s^{1},\dots, s^{\rho}+h,\dots,s^{d-1},t)-\psi(s,t)}{h}
$$
is the $\rho$th difference quotient of size $h>0$, see, e.g., \cite[Sec.~5.8.2]{evans2}. From the equation (\ref{f}) rewritten in the form sense for the test function $-\partial^{-h}_{\rho}\partial^{h}_{\rho}\psi$, one can obtain
\begin{equation}
    \begin{aligned}
        \int_{\Omega} \overline{a}\frac{\partial^{2}\overline{\psi}}{\partial s^{\mu}\partial s^{\rho}}g^{\mu\nu}\frac{\partial^{2}\psi}{\partial s^{\nu}\partial s^{\rho}}\mathrm{d}\Sigma\wedge\mathrm{d}t &+ I_{\psi}\ \\
        &-\int_{\Omega}\frac{\partial^{2}\overline{\psi}}{\partial s^{\rho}\partial s^{\rho}}\psi\hspace{0.05cm}\mathrm{d}\Sigma\wedge\mathrm{d}t=-\int_{\Omega}\frac{\partial^{2}\overline{\psi}}{\partial s^{\rho}\partial s^{\rho}}P_{0}F\hspace{0.05cm}\mathrm{d}\Sigma\wedge\mathrm{d}t
    \end{aligned}
    \label{III}
\end{equation}
using the analogue of the integration by parts for the difference quotients and taking the limit $h\rightarrow 0$. See, e.g., \cite[proof~of~Thm.~5.8.3]{evans2} for more details on the ``integration by parts'' rule. The term $I_{\psi}$ is given by
$$
I_{\psi}:=\int_{\Omega} \frac{\partial^{2}\overline{\psi}}{\partial s^{\mu}\partial s^{\rho}}\left(g^{\mu\nu}\frac{\partial\overline{a}}{\partial s^{\rho}}+\frac{\partial g^{\mu\nu}}{\partial s^{\rho}}\overline{a}+g^{\mu\nu}\overline{a}\frac{\partial \left|g\right|^{1/2}}{\partial s^{\rho}}\left|g\right|^{-1/2}\right)\frac{\partial\psi}{\partial s^{\nu}}\mathrm{d}\Sigma\wedge\mathrm{d}t.
$$
Following estimates can be proven due to the assumptions, (\ref{metrika}) and (\ref{4}):  
\begin{align}
\left| \int_{\Omega}\frac{\partial^{2}\overline{\psi}}{\partial s^{\rho}\partial s^{\rho}}P_{0}F\hspace{0.05cm}\mathrm{d}\Sigma\wedge\mathrm{d}t\right|&\leq \delta\left\| \frac{\partial^{2}\psi}{\partial s^{\rho}\partial s^{\rho}}\right\|_{0}^{2}+\delta^{-1}\| F\|_{0}^{2},\
\label{Ia}
\\
\left| \int_{\Omega}\frac{\partial^{2}\overline{\psi}}{\partial s^{\rho}\partial s^{\rho}}\psi\hspace{0.05cm}\mathrm{d}\Sigma\wedge\mathrm{d}t\right| &\leq \delta\left\| \frac{\partial^{2}\psi}{\partial s^{\rho}\partial s^{\rho}}\right\|_{0}^{2}+\delta^{-1}\| \psi\|_{0}^{2}\leq \delta\left\| \frac{\partial^{2}\psi}{\partial s^{\rho}\partial s^{\rho}}\right\|_{0}^{2}+C\delta^{-1}\| F\|_{0}^{2},\
\label{Ib}
\\
\left|I_{\psi}\right| &\leq \delta\left\| \frac{\partial^{2}\psi}{\partial s^{\mu}\partial s^{\rho}}\right\|_{0}^{2}+C\delta^{-1}\left\| \frac{\partial\psi}{\partial s^{\nu}}\right\|_{0}^{2}\leq \delta\left\| \frac{\partial^{2}\psi}{\partial s^{\mu}\partial s^{\rho}}\right\|_{0}^{2}+C\delta^{-1}\| F\|_{0}^{2}.\
\label{Ic}
\end{align}
The first terms of the final estimates can be treated as a perturbation of the first term of (\ref{III}) considering sufficiently small positive $\delta$. Additionally, we can gain
\begin{equation}
   \int_{\Omega} \overline{a}\frac{\partial^{2}\overline{\psi}}{\partial s^{\mu}\partial s^{\rho}}g^{\mu\nu}\frac{\partial^{2}\psi}{\partial s^{\nu}\partial s^{\rho}}\mathrm{d}\Sigma\wedge\mathrm{d}t\geq C\left\| \frac{\partial^{2}\psi}{\partial s^{\mu}\partial s^{\rho}}\right\|^{2}
    \label{Id}
\end{equation}
using the asumptions and (\ref{metrika}). Relations (\ref{Ia})--(\ref{Id}) directly implies 
$$
\left\| \frac{\partial^{2}\psi}{\partial s^{\mu}\partial s^{\rho}}\right\| \leq C\|F\|
$$
and the estimate (\ref{delta}) is shown.
\end{proof}
\begin{lemma}\label{l3}
Let $a_{\varepsilon} \in L^\infty(\Omega)$ 
satisfy both~(\ref{a1}) and (\ref{a3}) for almost every $(s,t) \in \Omega$ and all sufficiently small~$\eps$.
If $\left\{\psi_{\varepsilon}\right\}_{\varepsilon> 0}$ 
and~$\psi$ are defined by (\ref{c}) and (\ref{d})
for an arbitrary $G\in L^{2}\left(\Omega\right)$
and $F\in L^{2}\left(\Omega\right)$,
respectively, 
then the inequality
\begin{equation}
    \label{p}
    \left| \int_{\Omega} a_{\varepsilon} \frac{\partial \overline{\psi}}{\partial s^{\mu}}G^{\mu \nu}\frac{\partial P_{0}^{\perp} \psi_\varepsilon}{\partial s^{\nu}} f_{\varepsilon}^{1/2}\hspace{0.1cm} \mathrm{d}\Sigma\wedge\mathrm{d}t \right|  \leq C \varepsilon \|F\| \left\| G \right\|
\end{equation}
is valid for all sufficiently small~$\eps$.
\end{lemma}
\begin{proof}
Let us write
$$
 \int_{\Omega} a_{\varepsilon} \frac{\partial \overline{\psi}}{\partial s^{\mu}}G^{\mu \nu}\frac{\partial P_{0}^{\perp}\psi_\varepsilon}{\partial s^{\nu}} f_{\varepsilon}^{1/2} \mathrm{d}\Sigma\wedge\mathrm{d}t =  I_{1} + I_{2} 
$$ 
with
\begin{equation*}
\begin{aligned}
I_{1} &:= \int_{\Omega} a_{\varepsilon}  \frac{\partial \overline{\psi}}{\partial s^{\mu}}(G^{\mu \nu}-g^{\mu\nu})\frac{\partial P_{0}^{\perp}\psi_\varepsilon}{\partial s^{\nu}} f_{\varepsilon}^{1/2}\hspace{0.1cm} \mathrm{d}\Sigma\wedge\mathrm{d}t \,,
\\
I_{2} &:= \int_{\Omega} a_{\varepsilon} \frac{\partial \overline{\psi}}{\partial s^{\mu}}g^{\mu \nu}\frac{\partial P_{0}^{\perp}\psi_\varepsilon}{\partial s^{\nu}} f_{\varepsilon}^{1/2} \mathrm{d}\Sigma\wedge\mathrm{d}t \,.
\end{aligned}
\end{equation*}
Using $G^{\mu \nu}-g^{\mu\nu} = O(\eps)$ as $\eps\mapsto 0$, the integral $I_{1}$ can be estimated as follows
\begin{equation}\label{J1}
  |I_1| \leq C \eps \left\| a_{\varepsilon}\right\|_{\infty} \| | \nabla_g \psi |_{g} \| \left\| | \nabla_g \left(P_{0}^{\perp}\psi_\eps\right)|_{g}\right\| 
  \leq C \eps \|F\| \|G\| 
  \,.
\end{equation}
\indent Due to (\ref{delta}), $\Delta_{g}\psi\in \mathcal{H}_{0}$. Consequently, the integration by parts yields the first line. The final estimation follows from Lemmas \ref{ch}, \ref{nabla} and (\ref{4}). 
\begin{equation}\label{J2}
\begin{aligned}
   | I_{2} | &= \left| \int_{\Omega}\big[ a_{\varepsilon} (\Delta_{g}\psi) P_{0}^{\perp}\psi_\varepsilon f_{\varepsilon}^{1/2} +  \frac{\partial \overline{\psi}}{\partial s^{\mu}}g^{\mu \nu}\frac{\partial a_\varepsilon}{\partial s^{\nu}} P_{0}^{\perp}\psi_{\varepsilon} f_{\varepsilon}^{1/2}+ a_{\varepsilon} \frac{\partial \overline{\psi}}{\partial s^{\mu}}g^{\mu \nu}\frac{\partial f_{\varepsilon}^{1/2}}{\partial s^{\nu}} P_{0}^{\perp}\psi_\varepsilon \big] \mathrm{d}\Sigma\wedge\mathrm{d}t \right| \
   \\
   &\leq C \left[ \left\| a_{\varepsilon}\right\|_{\infty}\left\| \Delta_{g}\psi \right\| +\left( \left\| \left| \nabla_{g}a_{\varepsilon}\right|_{g} \right\|_{\infty} +\left\| \left| \nabla_{g}(f_{\varepsilon}^{1/2})\right|_{g} \right\|_{\infty}\right)  \left\| \left| \nabla_{g}\psi \right|_{g} \right\| \right] \left\| P_{0}^{\perp}\psi_{\varepsilon}\right\| 
   \\
   &\leq C \eps \|F\| \|G\| 
  \,
\end{aligned}    
\end{equation}
Combining (\ref{J1}) and (\ref{J2}), the estimate (\ref{p}) is proven.
\end{proof}
\subsection{The main result and its proof}\label{NRR}
Now we are in a position to state the main result of this paper.
Recall the notations $\left<a_{\varepsilon}\right>$ and $\overline{a}$
introduced in~(\ref{d0}).

\begin{theorem}\label{MTHM}
Let $a_{\varepsilon} \in L^\infty(\Omega)$ be satisfying~(\ref{a1}) and (\ref{a3}). 
In addition, assume  
\begin{equation}\label{a4}
  d(\varepsilon):=
  \left\|
  \left<{a}_{\varepsilon}\right> - 2 \overline{a} 
  \right\|_{L^\infty(\Sigma)}
  \xrightarrow[\varepsilon\mapsto 0]{} 0
  .
\end{equation}
 Then there exists a positive constant~$C$ such that
\begin{equation*}
    \label{rk}
    \left\| U_{\varepsilon}\left(H_{\varepsilon}+1\right)^{-1}U_{\varepsilon}^{-1}-  \left(H_{\mathrm{eff}}+1\right)^{-1}\oplus 0^{\perp} \right\|_{0\rightarrow 0} \leq C \max\left\{ \varepsilon, d(\varepsilon)\right\}
\end{equation*}
for all sufficiently small~$\eps$,
where $0^{\perp}$ denotes the zero operator on $\mathcal{H}_{0}^{\perp}$
and $\| \cdot\|_{0\rightarrow 0}$ stands for 
the operator norm in $L^2(\Omega)$.
\end{theorem}
\begin{proof}
Due to relations (\ref{c}) and (\ref{d}), one has 
$$
\begin{aligned}
\lefteqn{
\left(F,\left(U_{\varepsilon}\left(H_{\varepsilon}+1\right)^{-1}U_{\varepsilon}^{-1}-  \left(H_{\mathrm{eff}}+1\right)^{-1}\oplus 0^{\perp}\right)G\right)_{0}
}
\\
&=\left(F,U_{\varepsilon}\left(H_{\varepsilon}+1\right)^{-1}U_{\varepsilon}^{-1}G\right)_{0}
-\left(\left(\left(H_{\mathrm{eff}}+1\right)^{-1}\oplus 0^{\perp}\right)F,G\right)_{0}
\\
&=\left(\left(P_{0}+P_{0}^{\perp}\right)F,\left(P_{0}+P_{0}^{\perp}\right)U_{\varepsilon}\psi_{\varepsilon}\right)_{0}-\left(\left(H_{\mathrm{eff}}+1\right)^{-1}P_{0}F,G\right)_{0}
\\
&=\left(\left(H_{\mathrm{eff}}+1\right)\psi,P_{0}(U_{\varepsilon}\psi_{\varepsilon})\right)_{0}+\left(P_{0}^{\perp}F,P_{0}^{\perp}(U_{\varepsilon}\psi_{\varepsilon})\right)_{0}-\left(\psi,U_{\varepsilon}\left(H_{\varepsilon}+1\right)\psi_{\varepsilon}\right)_{0}
\\
& =\left(H_{\mathrm{eff}}\psi,P_{0}(U_{\varepsilon}\psi_{\varepsilon})\right)_{0}+\left(P_{0}^{\perp}F,P_{0}^{\perp}(U_{\varepsilon}\psi_{\varepsilon})\right)_{0}-\left(U_{\varepsilon}^{-1}\psi,H_{\varepsilon}\psi_{\varepsilon}\right)_{\varepsilon}
\\
&= {h}_{\mathrm{eff}}\left(\psi,P_{0}(U_{\varepsilon}\psi_{\varepsilon})\right)-h_{\varepsilon}\left(U_{\varepsilon}^{-1}\psi,\psi_{\varepsilon}\right)+ \left(P_{0}^{\perp}F,P_{0}^{\perp}(U_{\varepsilon}\psi_{\varepsilon})\right)_{0} =: T,
\end{aligned}
$$
where $Q_{\mathrm{eff}}$ and $Q_{\varepsilon}$ are sesquilinear forms of $H_{\mathrm{eff}}$ and $H_{\varepsilon}$. 
The last line explicitly equals
 $$
 \begin{aligned}
 T &= \int_{\Omega}\left[ \overline{a}\frac{\partial\overline{\psi}}{\partial s^{\mu}}g^{\mu \nu}\frac{\partial P_{0}\left(f_{\varepsilon}^{1/2}\psi_{\varepsilon}\right)}{\partial s^{\nu}}-a_{\varepsilon}\frac{\partial \left( f_{\varepsilon}^{-1/2}\overline{\psi}\right)}{\partial s^{\mu}}G^{\mu \nu}\frac{\partial \psi_{\varepsilon}}{\partial s^{\nu}} f_{\varepsilon}\right] \mathrm{d}\Sigma\wedge\mathrm{d}t
\\ &\quad
 -\frac{1}{\varepsilon^{2}}\int_{\Omega}a_{\varepsilon}\overline{\psi}\frac{\partial\left( f_{\varepsilon}^{-1/2}\right)}{\partial t}\frac{\partial\psi_{\varepsilon}}{\partial t}f_{\varepsilon}\hspace{0.1cm}\mathrm{d}\Sigma\wedge\mathrm{d}t +\int_{\Omega} \overline{P_{0}^{\perp}F}P_{0}^{\perp}(\psi_{\varepsilon}f_{\varepsilon}^{1/2})\hspace{0.1cm}\mathrm{d}\Sigma\wedge\mathrm{d}t
  \end{aligned}
 $$

 Let us write
$$
J := \int_{\Omega}\overline{a}\frac{\partial\overline{\psi}}{\partial s^{\mu}}g^{\mu \nu}\frac{\partial P_{0}\left(f_{\varepsilon}^{1/2}\psi_{\varepsilon}\right)}{\partial s^{\nu}}\mathrm{d}\Sigma\wedge\mathrm{d}t-a_{\varepsilon}\frac{\partial \left( f_{\varepsilon}^{-1/2}\overline{\psi}\right)}{\partial s^{\mu}}G^{\mu \nu}\frac{\partial\psi_{\varepsilon}}{\partial s^{\nu}} f_{\varepsilon} \hspace{0.1cm} \mathrm{d}\Sigma\wedge\mathrm{d}t=  J_{1} + J_{2} + J_{3} - J_{4} - J_{5} 
$$ 
with
\begin{equation*}
\begin{aligned}
J_{1} &:= \int_{\Omega}\overline{a}\frac{\partial\overline{\psi}}{\partial s^{\mu}}g^{\mu \nu}\left(\frac{\partial P_{0}\left(f_{\varepsilon}^{1/2}\psi_{\varepsilon}\right)}{\partial s^{\nu}}-\frac{\partial P_{0}\psi_{\varepsilon}}{\partial s^{\nu}}\right)\mathrm{d}\Sigma\wedge\mathrm{d}t \,,
\\
J_{2} &:= \int_{\Omega}\left(\overline{a}-a_{\varepsilon}f_{\varepsilon}\right)\frac{\partial \overline{\psi}}{\partial s^{\mu}}g^{\mu \nu}\frac{\partial P_{0} \psi_{\varepsilon}}{\partial s^{\nu}}\hspace{0.1cm}\mathrm{d}\Sigma\wedge\mathrm{d}t \,,
\end{aligned}
\end{equation*}
\begin{equation*}
\begin{aligned}
J_{3} &:= \int_{\Omega}a_{\varepsilon}\frac{\partial \overline{\psi}}{\partial s^{\mu}}\left(g^{\mu \nu}-f_{\varepsilon}^{-1/2}G^{\mu \nu}\right)\frac{\partial P_{0} \psi_{\varepsilon}}{\partial s^{\nu}}\hspace{0.1cm}f_{\varepsilon}\hspace{0.1cm}\mathrm{d}\Sigma\wedge\mathrm{d}t \,,
\\
J_{4} &:= \int_{\Omega}a_{\varepsilon} \frac{\partial f_{\varepsilon}^{-1/2}}{\partial s^{\mu}}\overline{\psi}G^{\mu \nu}\frac{\partial \psi_{\varepsilon}}{\partial s^{\nu}} f_{\varepsilon}\hspace{0.1cm}\mathrm{d}\Sigma\wedge\mathrm{d}t \,,
\\
J_{5} &:= \int_{\Omega}a_{\varepsilon}f_{\varepsilon}^{1/2}\frac{\partial \overline{\psi}}{\partial s^{\mu}}G^{\mu \nu}\frac{\partial P_{0}^{\perp} \psi_{\varepsilon}}{\partial s^{\nu}} \hspace{0.1cm}\mathrm{d}\Sigma\wedge\mathrm{d}t \,.
\end{aligned}
\end{equation*}
 \noindent In order to estimate the integral $J_{1}$, we show the following equalities using the properties of the orthogonal projection $P_{0}$ and the exchange of derivatives and integral.
$$
\begin{aligned}
\frac{\partial P_{0}\left(f_{\varepsilon}^{1/2}\psi_{\varepsilon}\right)}{\partial s^{\nu}}
&=\frac{\partial}{\partial s^{\nu}}\left[\frac{1}{2} \int_{-1}^{1}\left(f_{\varepsilon}^{1/2}-1\right)\left(P_{0}+P_{0}^{\perp}\right)\psi_{\varepsilon}\mathrm{d}t+P_{0}\psi_{\varepsilon} \right]
\\
&= \frac{\partial P_{0}\psi_{\varepsilon}}{\partial s^{\nu}}+\frac{1}{2}\frac{\partial P_{0}\psi_{\varepsilon}}{\partial s^{\nu}}\int_{-1}^{1}\left(f_{\varepsilon}^{1/2}-1\right)\mathrm{d}t+ \frac{1}{2}P_{0}\psi_{\varepsilon}\int_{-1}^{1}\frac{\partial f_{\varepsilon}^{1/2}}{\partial s^{\nu}}\mathrm{d}t
\\
& \quad
+\frac{1}{2}\int_{-1}^{1}\frac{\partial f_{\varepsilon}^{1/2}}{\partial s^{\nu}}P_{0}^{\perp}\psi_{\varepsilon}\mathrm{d}t+\frac{1}{2}\int_{-1}^{1}\left(f_{\varepsilon}^{1/2}-1\right)\frac{\partial P_{0}^{\perp}\psi_{\varepsilon}}{\partial s^{\nu}}\mathrm{d}t.
\end{aligned}
$$ 
Consequently, the integral $J_{1}$ satisfies 
$$
\begin{aligned}
|J_{1}|
&\leq C \left\| \left| \nabla_{g}\psi \right|_{g}\right\| \bigg[
 \left\| \left| \nabla_{g}f_{\varepsilon}^{1/2}\right|_{g}\right\|_{\infty}\Big(
\left\| P_{0} \psi_{\varepsilon} \right\|
+\left\| P_{0}^{\perp} \psi_{\varepsilon} \right\| \Big)
\\
 & \quad + \left\|f_{\varepsilon}^{1/2}-1\right\|_{\infty} \left( \left\| \left| \nabla_{g}\left(P_{0}\psi_{\varepsilon}\right) \right|_{g}\right\| + \left\| \left| \nabla_{g}\left(P_{0}^{\perp}\psi_{\varepsilon}\right) \right|_{g}\right\| \right)\bigg]
 \\
 & \leq \varepsilon\hspace{0.05cm} C \left\|F \right\| \left\|G \right\|,
\end{aligned}
$$
where the last line is valid due to Lemmas \ref{odhady} and \ref{ch}. Furthermore, integrals $J_{2}$, $J_{3}$ and $J_{4}$ comply with estimates
$$
\begin{aligned}
\left| J_{2} \right|
&=
\left| \int_{\Sigma}\left(\int_{-1}^{1}\overline{a}-a_{\varepsilon}f_{\varepsilon}\mathrm{d}t\right)\frac{\partial \overline{\psi}}{\partial s^{\mu}}g^{\mu \nu}\frac{\partial P_{0} \psi_{\varepsilon}}{\partial s^{\nu}}\hspace{0.1cm}\mathrm{d}\Sigma \right| \leq d(\varepsilon) \hspace{0.05cm}C  \left\| \left| \nabla_{g}\psi\right|_{g}\right\| \left\| \left| \nabla_{g}\left( P_{0}\psi_{\varepsilon}\right) \right|_{g}\right\|
 \\
 &\leq d(\varepsilon) \hspace{0.05cm} C \left\|F \right\| \left\|G \right\| \,,
 \\
 \left| J_{3} \right| &\leq \left| \int_{\Omega}a_{\varepsilon}\frac{\partial \overline{\psi}}{\partial s^{\mu}}\left(f_{\varepsilon}^{-1/2}-1\right)G^{\mu \nu}\frac{\partial P_{0} \psi_{\varepsilon}}{\partial s^{\nu}}\hspace{0.1cm}f_{\varepsilon}\hspace{0.1cm}\mathrm{d}\Sigma\wedge\mathrm{d}t \right| + \left| \int_{\Omega}a_{\varepsilon}\frac{\partial \overline{\psi}}{\partial s^{\mu}}\left(G^{\mu \nu}-g^{\mu \nu}\right)\frac{\partial P_{0} \psi_{\varepsilon}}{\partial s^{\nu}}\hspace{0.1cm}f_{\varepsilon}\hspace{0.1cm}\mathrm{d}\Sigma\wedge\mathrm{d}t \right|
\\
& \leq \varepsilon\hspace{0.05cm}C  \left\| \left| \nabla_{g}\psi\right|_{g}\right\| \left\| \left| \nabla_{g}\left( P_{0}\psi_{\varepsilon}\right) \right|_{g}\right\|  \leq \varepsilon\hspace{0.05cm} C \left\|F \right\|\left\|G \right\|,
\\
\left| J_{4} \right| &\leq C  \left\| \left| \nabla_{g}\left(f_{\varepsilon}^{-1/2}\right)\right|_{g}\right\|_{\infty} \left\| \left| \nabla_{g} \psi_{\varepsilon} \right|_{g}\right\|  \left\| \psi \right\| \leq \varepsilon\hspace{0.05cm} C \left\|F \right\| \left\|G \right\|,
\end{aligned}
 $$
 which are proven using $G^{\mu \nu}-g^{\mu\nu} = O(\eps)$ as $\eps\mapsto 0$ and Lemma \ref{odhady}. From the estimates of $J_{1}$--$J_{4}$ mentioned above and Lemma \ref{l3}, 
 it can be seen that
\begin{equation}\label{J}
\left| J \right| \leq \max\left\{d(\varepsilon),\varepsilon\right\} \hspace{0.05cm} C \left\|F \right\| \left\|G \right\|.
\end{equation}
\indent In this part of proof, we will follow an idea in \cite[proof~of~Thm.~1.4]{kre-tus-2}. First, we rewrite the equation (\ref{rr}) for the function $t\hspace{0.02cm}K_{1}\psi$,
 explicitly 
 $$
 \begin{aligned}
     \int_{\Omega} a_{\varepsilon} t \hspace{0.02cm}K_{1}\frac{\partial \overline{\psi}}{\partial s^{\mu}}G^{\mu \nu}\frac{\partial \psi_\varepsilon}{\partial s^{\nu}} f_{\varepsilon}\hspace{0.1cm} \mathrm{d}\Sigma\wedge\mathrm{d}t+\int_{\Omega} a_{\varepsilon} t \hspace{0.03cm}\overline{\psi}\frac{\partial K_{1}}{\partial s^{\mu}}G^{\mu \nu}\frac{\partial \psi_\varepsilon}{\partial s^{\nu}} f_{\varepsilon}\hspace{0.1cm} \mathrm{d}\Sigma\wedge\mathrm{d}t &
     \\
 + \frac{1}{\varepsilon^{2}}\int_{\Omega} a_{\varepsilon}K_{1}\overline{\psi} \frac{\partial \psi_{\varepsilon}}{\partial t}f_{\varepsilon}\hspace{0.1cm} \mathrm{d}\Sigma\wedge\mathrm{d}t +\int_{\Omega} t\hspace{0.02cm}K_{1}\overline{\psi} \psi_{\varepsilon}f_{\varepsilon}\hspace{0.1cm} \mathrm{d}\Sigma\wedge\mathrm{d}t &=\int_{\Omega}t\hspace{0.02cm}K_{1} \overline{\psi}G f^{1/2}_{\varepsilon}\hspace{0.1cm} \mathrm{d}\Sigma\wedge\mathrm{d}t. 
  \end{aligned}
$$  
Consequently,
$$
\left|  \frac{1}{\varepsilon^{2}}\int_{\Omega} a_{\varepsilon}K_{1}\overline{\psi} \frac{\partial \psi_{\varepsilon}}{\partial t}f_{\varepsilon}\hspace{0.1cm} \mathrm{d}\Sigma\wedge\mathrm{d}t \right| \leq C \left\|F \right\| \left\|G \right\|. 
$$
Then the observation 
$$\frac{\partial f^{-1/2}}{\partial t}= \frac{d-1}{2}\varepsilon K_{1}+\mathcal{O}(\varepsilon^{2}) \qquad \mbox{as} \qquad
  \eps \to 0 $$
gives the following estimate for all sufficiently small $\varepsilon$:
\begin{equation}\label{K1}
\left| \frac{1}{\varepsilon^{2}}\int_{\Omega}a_{\varepsilon}\overline{\psi}\frac{\partial\left( f_{\varepsilon}^{-1/2}\right)}{\partial t}\frac{\partial\psi_{\varepsilon}}{\partial t}f_{\varepsilon}\hspace{0.1cm}\mathrm{d}\Sigma\wedge\mathrm{d}t\right| \leq \varepsilon\hspace{0.05cm} C \left\| F \right\| \left\|G \right\|.
\end{equation}
 \indent The last part can be estimated using Lemmas \ref{odhady} and \ref{ch} as follows:
\begin{equation}\label{final}
\begin{aligned}  
\left| \int_{\Omega} \overline{P_{0}^{\perp}F}P_{0}^{\perp}\left(\psi_{\varepsilon}f_{\varepsilon}^{1/2}\right)\hspace{0.1cm}\mathrm{d}\Sigma\wedge\mathrm{d}t \right| 
&\leq  C\left\| P_{0}^{\perp}F\right\| \left\| P_{0}^{\perp}\left(\psi_{\varepsilon}f_{\varepsilon}^{1/2}\right) \right\|
\\
&\leq  C\left\|F\right\| \left( \left\| P_{0}^{\perp}\left(\psi_{\varepsilon}f_{\varepsilon}^{1/2}-\psi_{\varepsilon}\right) \right\| + \left\| P_{0}^{\perp}\psi_{\varepsilon} \right\| \right)
\\
&\leq C \left\|F\right\| \left( \left\|f_{\varepsilon}^{1/2}-1\right\|_{\infty} \left\| \psi_{\varepsilon} \right\| + \left\| P_{0}^{\perp}\psi_{\varepsilon} \right\|\right) 
\\
&\leq \varepsilon\hspace{0.05cm} C \left\|F \right\| \left\|G \right\|.
\end{aligned}
\end{equation}
Relations \eqref{J}--\eqref{final} directly implies 
$$\left| \left(F,\left(U_{\varepsilon}\left(H_{\varepsilon}+1\right)^{-1}U_{\varepsilon}^{-1}-  \left(H_{\mathrm{eff}}+1\right)^{-1}\oplus 0^{\perp}\right)G\right)_{0} \right|\leq C \max\left\{ \varepsilon, d(\varepsilon)\right\}\hspace{0.05cm} \left\|F\right\|_{0} \left\|G \right\|_{0},$$
so the theorem is proven.
\end{proof}
\label{rate}

\section{Conclusion}\label{Sec.end}
Let us discuss the optimality of assumptions under which Theorem \ref{MTHM} holds. 

The present geometric assumptions follow Yachimura's setting~\cite{tosy},
but they can be weaken in several directions.
First, the $C^3$-smoothness of the hypersurface~$\Sigma$
could be reduced to the $C^2$-smoothness by employing 
the trick of~\cite{sed} of an $\eps$-dependent regularisation  
of curvatures in the unitary transform~\eqref{Ue}.
It is also possible to go beyond the $C^2$-smoothness
in the spirit of~\cite{KZ3}.
More importantly, there is no substantial reason 
for considering compact hypersurfaces. 
The generalised norm-resolvent convergence 
is expected to hold for tubular neighbourhoods 
of arbitrary hypersurfaces;
in the non-compact case, it is only important to assume
the injectivity hypothesis~\eqref{overlap2} \emph{ad hoc}
(see, e.g., \cite{kre-tus}).
As a consequence, one gets the convergence of eigenvalues
below the essential spectrum as well as of the corresponding eigenfunctions.
 
The present assumptions about the non-homogeneity~$a$
substantially generalise the piece-wise constant 
setting of Yachimura's~\cite{tosy}.
The basic hypothesis~\eqref{a0} is rather standard
for elliptic problems.
On the other hand, while the ``longitudinal'' differentiability
of the non-homogeneity~\eqref{a3} is certainly needed
for our approach based on the \emph{a priori} estimate of Lemma~\ref{nabla},
we make no claim about its necessity 
for the validity of the norm-resolvent convergence.
Similarly, we leave open the question of optimality of hypothesis~\eqref{a4},
while the existence of a pointwise limit~$\overline{a}$
is necessary for the very definition of 
the effective Hamiltonian~$H_{\mathrm{eff}}$.

Finally, let us mention interesting extentions of the studied model.
First, it is possible to consider more general boundary conditions,
namely Robin boundary conditions, 
including those with a complex coupling~\cite{BK2,kre-sieg},
as well as combined boundary conditions in the spirit of
\cite{DKriz2,K5,K10}.
Second, motivated by a recent development of quasi-Hermitian
quantum mechanics~\cite{KSTV},  
the non-homogeneity~$a$ could be considered to be a complex function,
possibly enriched by complex magnetic fields~\cite{K13}.

\subsection*{Acknowledgment}
The author was supported
by the EXPRO grant No.~20-17749X
of the Czech Science Foundation.
 
%
\bibliography{main}
\bibliographystyle{amsplain}

\end{document}